\documentclass{article}
\usepackage[utf8]{inputenc}
\usepackage[T1]{fontenc}

\usepackage[margin=1.2in]{geometry}
\usepackage{authblk}
\usepackage{amsmath}
\usepackage{amssymb}
\usepackage{amsthm}
\usepackage{mathtools}
\usepackage{hyperref}
\usepackage[capitalise]{cleveref}
\usepackage{doi}
\usepackage{tikz}
\usepackage{graphicx}

\usepackage[style=alphabetic, maxalphanames=4, maxbibnames=10]{biblatex}
\bibliography{main}

\newtheorem{theorem}{Theorem}
\newtheorem{corollary}[theorem]{Corollary}
\newtheorem{lemma}[theorem]{Lemma}
\newtheorem{proposition}{Proposition}
\newtheorem{definition}{Definition}

\DeclarePairedDelimiter{\ceil}{\lceil}{\rceil}

\makeatletter
\def\blfootnote{\gdef\@thefnmark{}\@footnotetext}
\makeatother

\title{On speeding up factoring with quantum SAT solvers}
\author[1]{Michele Mosca}
\author[2]{João Marcos Vensi Basso\thanks{Corresponding author: \href{mailto:Joao.Vensi_Basso@tufts.edu}{Joao.Vensi\_Basso@tufts.edu}}}
\author[3]{Sebastian R. Verschoor}

\affil[1,3]{Institute for Quantum Computing, University of Waterloo, Canada}
\affil[2]{Department of Physics and Astronomy, Tufts University, USA}
\affil[1]{Department of Combinatorics \& Optimization, University of Waterloo, Canada}
\affil[1]{Perimeter Institute for Theoretical Physics, Waterloo, Canada}
\affil[1]{evolutionQ Inc., Waterloo, Canada}
\affil[3]{David R. Cheriton School of Computer Science, University of Waterloo, Canada}
\date{\today}

\begin{document}

\maketitle
\blfootnote{Author list in alphabetical order; see~\url{https://www.ams.org/profession/leaders/culture/CultureStatement04.pdf}.}

\begin{abstract}
    There have been several efforts to apply quantum SAT solving methods to factor large integers. While these methods may provide insight into quantum SAT solving, to date they have not led to a convincing path to integer factorization that is competitive with the best known classical method, the Number Field Sieve. Many of the techniques tried involved directly encoding multiplication to SAT or an equivalent NP-hard problem and looking for satisfying assignments of the variables representing the prime factors. The main challenge in these cases is that, to compete with the Number Field Sieve, the quantum SAT solver would need to be superpolynomially faster than classical SAT solvers. In this paper the use of SAT solvers is restricted to a smaller task related to factoring: finding smooth numbers, which is an essential step of the Number Field Sieve. We present a SAT circuit that can be given to quantum SAT solvers such as annealers in order to perform this step of factoring. If quantum SAT solvers achieve any speedup over classical brute-force search, then our factoring algorithm is faster than the classical NFS.
\end{abstract}

\section{Introduction}

Factoring integers by translating the problem directly into a satisfiability (SAT)
instance or any equivalent NP-hard problem does not appear to be efficient, even when quantum solvers are assumed to be able to achieve a quadratic speedup~\cite{mv19}.
More importantly, the strategy does not even appear to perform better than the best known
classical method: the Number Field Sieve (NFS)~\cite{nfs}.

A subroutine of the NFS is to search for $y$-smooth numbers of a particular form, where an integer is $y$-smooth if all of its prime factors are $\leq y$.
Using Grover's algorithm~\cite{grover96} this search can be done faster,
so that a speedup over the classical method is achieved~\cite{bbm17}.
Although the resulting algorithm runs in super-polynomial time (and is thus slower
than Shor's algorithm~\cite{shor94}), it requires asymptotically fewer logical qubits to implement.

We investigate the strategy of replacing Grover's search in the described low-resource algorithm
by translating the smooth detection process into a satisfiability instance to be evaluated by a SAT solver. While the low-resource algorithm in~\cite{bbm17} requires a fault-tolerant quantum computer, one can alternatively attempt to solve these SAT instances with any quantum SAT solving algorithm or heuristic, such as a quantum annealer. \emph{If} the quantum SAT solving heuristic achieves a speed-up over classical SAT solving algorithms, then we show that this leads to a factoring algorithm that is asymptotically faster than the regular NFS.
While there is no convincing evidence to date that non-fault-tolerant quantum SAT solvers will provide an asymptotic speed-up over classical SAT solvers, with this approach we at least avoid the situation where the quantum SAT solver must outperform classical SAT solvers by a superpolynomial factor in order to compete with the NFS.

We first demonstrate theoretically that some approaches to translating smoothness testing to a SAT instance are too expensive. In practice, one might hope that SAT solvers would be able to pick up on patterns of these specific circuits and to achieve potential speedups. Note that this does not appear to happen in the direct factoring strategy~\cite{mv19}, but there is the possibility that it would for the more specific problem of smooth number detection. After implementing one of the circuits, however, the benchmarks suggest that this is not the case.

We found one circuit implementing the Elliptic Curve Method (ECM)~\cite{len87} which, if used as a subroutine of the NFS, could result in a speedup for factoring integers
since
SAT solvers can use this circuit to find smooth numbers asymptotically as fast as brute-force search. Moreover, if quantum annealers or other SAT solvers achieve any speedup over such classical SAT solving, then our algorithm is faster than the classical one. In the optimistic case that quantum SAT solvers achieve a full quadratic speedup, then our algorithm has the same time complexity as the low-resource quantum algorithm of~\cite{bbm17}. 

\subsection{Contributions of this paper}
We show in general that a few approaches for smoothness detection with SAT circuits are not enough to speed up the NFS. Moreover, we run benchmarks and find that a classical SAT solver does not appear to pick up on any patterns that allow one to claim otherwise for these approaches. Most importantly, we present a circuit that, when used as a NFS subroutine, yields an algorithm with the same asymptotic runtime as the classical NFS, and faster if quantum SAT solvers achieve any non-trivial speedup. In the optimistic case that a quantum SAT solver achieves a full quadratic speedup, the algorithm would be as fast as the low-resource quantum algorithm, while not necessarily requiring a fault-tolerant quantum computer to operate.

\subsection{Nomenclature}
We refer to the algorithm in~\cite{nfs} as the classical NFS, to the algorithm in~\cite{bbm17} as the low-resource quantum NFS, and to our algorithm that uses ECM for smoothness detection as circuit-NFS.

\subsection{Organization}
In \Cref{previous_work}, we review previous work related to SAT solving as well as factoring, including a factoring algorithm that encodes a multiplication circuit as a SAT instance whose solution represents the prime factors. We also recall work done to speed up the Number Field Sieve using Grover's search on a quantum computer with $(\log N)^{2/3+o(1)}$ logical qubits, where $N$ is the number being factored. In \Cref{circuits_section}, we present a few encodings of smoothness detection into circuits. We show in general that the SAT instances belonging to smoothness-detection circuits which have the prime exponents or the factors of the number being tested for smoothness as variables cannot be solved fast enough to speed up factoring. Most importantly, we present a circuit implementing the ECM, analyze it and make statements about the solver runtime relative to the speedup obtained by a quantum SAT solver. Lastly, in \Cref{conclusion}, we discuss the results of this paper as well as future work.

\section{Previous Work}\label{previous_work}
The work in~\cite{mv19} investigates the use of SAT solvers for factoring semi-primes, that is, numbers with only two primes factors of similar size. It encodes a multiplication circuit into a SAT instance, fixing the output as the number being factored and making the multiplicands variable. Therefore, solving such SAT instance is equivalent to factoring the semi-prime. The paper finds no evidence that this approach to factoring via classical SAT solvers provides any advantage, or even matches the classical NFS. It also points out that quantum SAT solvers are not expected to do much better if factoring is encoded as a SAT instance in this direct fashion.

The general number field sieve (NFS) improves on the special number field sieve~\cite{snfs} by removing any restrictions on the numbers that can be factored. The NFS algorithm is conjectured to factor any integer $N$ in $L_N[1/3, (64/9)^{1/3} + o(1)]$ time, where $L_x[a,b] = \exp{(b(\log x)^{a} (\log \log x)^{1-a})}$ and $o(1) \rightarrow 0$ as $N \rightarrow \infty$. Here we give a brief overview of the algorithm, highlighting the details relevant to the present paper. Note that the NFS is explained and analyzed in thorough detail in \cite{nfs}. For a simplified overview, see~\cite[Section 2]{bbm17}, whose notation we follow.

The algorithm takes in an integer $N$ to be factored and parameters $d, y, u$, with
\begin{itemize}
  \item $y \in L^{\beta + o(1)}$
  \item $u \in L^{\epsilon + o(1)}$
  \item $d \in (\delta + o(1))(\log N)^{1/3}(\log \log N)^{-1/3}$
\end{itemize}
where $N > 2^{d^2}$, $L = L_N[1/3,1]$ and $\beta, \delta, \epsilon$ are parameters to be optimized for in the analysis.
 Further, define
\begin{itemize}
  \item $m := \lfloor N^{1/d} \rfloor$
  \item $U := \{(a,b) \in \mathbb{Z}^2:  \text{gcd}\{a,b\}=1, |a| \leq u, 0 < b \leq u \}$
  \item $f(X) := \sum_{i=0}^{d}c_i X^i$ where the $c_i$ are obtained by writing $N$ in base $m$: $N = \sum_{i=0}^{d}c_i m^i$
  \item $\alpha$ such that $f(\alpha)=0$
  \item $g(a,b) := (-b)^d f(-a/b) = \sum_{i=0}^d c_i a^{i} (-b)^{d-i}$ 
  \item $F(a,b) := (a+bm)g(a,b)$
  \item $\phi : \mathbb{Z}\big[ \alpha \big] \rightarrow \mathbb{Z}/N\mathbb{Z} : \sum_i a_i \alpha^i \rightarrow \sum_i a_i m^i$, a homomorphism.
\end{itemize}

From the above, one can see that $d$ represents the degree of the polynomial $f$ and that $u$ is, in a sense, a bound on the search space $U$. Moreover, as explained below, $y$ is taken to be the smoothness bound on $F(a,b)$. $N$ is assumed to be odd. 

The NFS attempts to find a suitable set $S \subseteq U$ such that on the rational side
\begin{equation}\label{rational_sq}
    \prod_{(a,b) \in S} (a+bm) = X^2 \text{ is a square in } \mathbb{Z}
\end{equation}
and on the algebraic side
\begin{equation}\label{algebraic_sq}
    f'(\alpha)^2 \prod_{(a,b) \in S} (a+b\alpha) = \beta^2 \text{ is a square in } \mathbb{Z}\big[ \alpha \big].
\end{equation}
The algorithm then outputs
\begin{equation}
    \text{gcd}\{N, \phi(\beta) - f'(m)X\}.
\end{equation}

In order to find an appropriate $S$, the algorithm looks for a $T \subseteq U$ such that $T = \{(a,b) \in \mathbb{Z}^2 : \text{gcd}\{a, b\}=1,|a| \leq u, 0 < b \leq u, F(a,b) \text{ is } y \text{-smooth}\}$, with $\#T \in y^{1+o(1)}$. After $T$ is found, a linear dependence relation between the exponent vectors (reduced modulo 2) of $F(a,b)$ for $(a,b) \in T$  reveals a suitable set $S \subseteq T$ such that both \Cref{rational_sq} and \Cref{algebraic_sq} are satisfied.

The two main bottlenecks of NFS are to (i) find T and (ii) find the linear dependence relation. In the classical NFS (i) takes $L^{2\epsilon + o(1)}$ time, since that is the size of $U$, and (ii) takes $L^{2\beta + o(1)}$ with Wiedemann's algorithm~\cite{wie86}. By balancing both, one obtains a total runtime of $L^{1.923}$. The low-resource algorithm does (i) using Grover's search and yields a better runtime, namely $L^{1.387}$. Note as well that, if (ii) is assumed to take $L^{2.5\beta + o(1)}$ as in~\cite{ber01}, the classical NFS ends up with runtime $L^{1.976}$ and the low-resource algorithm with $L^{1.456}$. For completeness we repeat the derivations in \Cref{gamma_1_corollary} and \Cref{gamma_2_corolary}.

\section{Circuits for smoothness detection}\label{circuits_section}

The circuit SAT problem asks whether there exists an input for a given Boolean circuit, encoded as a SAT instance, such that the output will be TRUE.
For a satisfiable circuit SAT formula in $v$ variables one can easily find a solution with $v$ queries to a decision oracle for SAT. In practice, the best known algorithms for deciding SAT implicitly also provide a solution and thus the repeated applications of a SAT decision algorithm are not necessary.
Using binary encoding for integers we construct circuits that encode a predicate on numbers, so that solving
the corresponding SAT instance is a search for numbers satisfying the predicate.
From here on we refer to this process as ``solving the circuit''.

Instead of using Grover's search to look for $(a,b) \in T$ as in~\cite{bbm17}, we let a SAT solver find these using the encoded circuit.
In particular, we encode the predicate ``$F(a,b)$ is a $y$-smooth number'' on the input pair $(a, b)$,
while we assume the conditions $|a| \leq u$ and $0 < b \leq u$ are enforced by the input encoding. Similar to~\cite{bbm17}, we assume that the case $\gcd\{a,b\} > 1$ is handled by post-processing.

A naive algorithm for circuit SAT simply evaluates the full circuit for every possible input until a one is found at the output.
For a circuit with $v$ input variables and size $g$ this strategy has runtime $O(2^v g)$, which is
the runtime we assume for solving circuits.
Given that circuit SAT is an NP-complete problem, it is widely believed that no efficient algorithm exists.
However, in practice modern SAT solvers perform well on solving large SAT instances for certain problems,
so that the conjectured runtime requires some confirmation in the form of benchmark results.

In this section we analyze a few natural circuits for implementing the required predicate and prove the approach does not offer any improvement over the classical NFS. We show that, in general, circuits encoding all primes $p_i\leq y$ or the prime exponents $e_i$ can not be solved efficient enough. On the other hand, solving a circuit implementing the Elliptic Curve Method (ECM)~\cite{len87} is shown to achieve runtimes comparable to that of the classical NFS. We recall a few results important for the analysis.

\begin{lemma}\label{lemma-sizes}
\begin{enumerate}
    \item $|(a+bm)| \leq 2 u N ^{1/d}$ and $|g(a,b)| \leq (d+1) N ^{1/d} u^{d}$
    \item $\log |F(a,b)| \in O((\log N)^{2/3}(\log \log N)^{1/3})$
    \item $\log \log |F(a,b)| \in O(\log \log N)$
\end{enumerate}
\end{lemma}

\begin{proof}
\begin{enumerate}
    \item Follows directly from the definitions of $g$ and $m$.
    \item $\log |F(a,b)| \leq \log 2(d+1) + \frac{2\log N}{d} + (d+1) \log u \in O(\frac{\log N}{d} + d \log u)$. Now $\frac{\log N}{d} + d \log u = (\epsilon \delta + \delta^{-1} + o(1))(\log N)^{2/3}(\log \log N)^{1/3} \subseteq O((\log N)^{2/3}(\log \log N)^{1/3})$.
    \item Taking logs of the expression above, the dominant term is $\log \log N$.
\end{enumerate}
\end{proof}

\begin{lemma}\label{lemma_omega_n}
If $F(a,b)$ is y-smooth, then $\Omega(F(a,b)) \in O((\log N)^{2/3}(\log \log N)^{1/3})$, where $\Omega(n)$ is the number of prime divisors of $n$ with multiplicity.
\end{lemma}

\begin{proof}
All the prime factors of $F(a,b)$ are at least $2$, so that $\Omega(F(a,b)) \geq \log_2 F(a,b)$. The result follows from \Cref{lemma-sizes}.
\end{proof}

\subsection{Circuit with variable exponents}

A natural idea is to hard-code all the primes $p_i \leq y$ into the circuit (see \Cref{fig:fullproduct}), and let $a, b$ and $e_i$ be the variables, where $1 \leq i \leq \pi(y)$, and $\pi(x)$ counts the number of primes $\leq x$. A satisfying assignment finds the exponent $e_i$ for each prime $p_i$ that forms the factorization of $F(a,b)$:
\begin{equation}\label{eqn:smooth}
    F(a,b) = \prod_{i=1}^{\pi(y)}p_i^{e_i}
\end{equation}

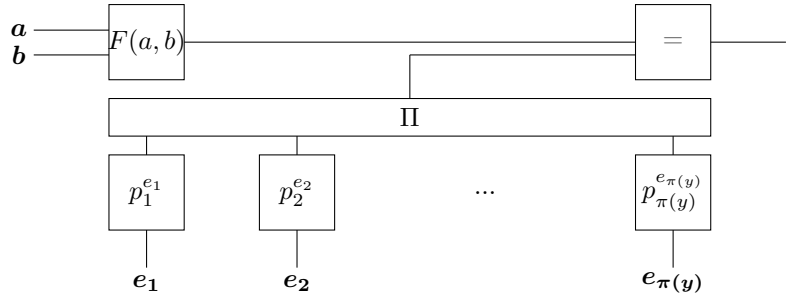
\begin{figure}[ht]
\centering
\begin{tikzpicture}
\draw (1,2.33) -- (2,2.33) node[pos=-0.2] {$\boldsymbol{b}$};
\draw (1,2.66) -- (2,2.66) node[pos=-0.2] {$\boldsymbol{a}$};
\draw (2,2) rectangle (3,3) node[pos=0.5] {$F(a,b)$};
\draw (3,2.5) -- (9,2.5);
\draw (9,2) rectangle (10,3) node[pos=0.5] {=};
\draw (10,2.5) -- (11,2.5);
\draw (2,1.25) rectangle (10,1.75) node[pos=0.5] {$\Pi$};
\draw (6,1.75) -- (6,2.33);
\draw (6,2.33) -- (9,2.33);
\draw (2.5,1) -- (2.5,1.25);
\draw (2,0) rectangle (3,1) node[pos=0.5] {$p_1^{e_1}$};
\draw (4.5,1) -- (4.5,1.25);
\draw (4,0) rectangle (5,1) node[pos=0.5] {$p_2^{e_2}$};
\draw (9.5,1) -- (9.5,1.25);
\draw (9,0) rectangle (10,1) node[pos=0.5] {$p_{\pi(y)}^{e_{\pi(y)}}$};
\node at (7,0.5) {...};
\draw (2.5,0) -- (2.5,-0.5); \node at (2.5,-0.7) {$\boldsymbol{e_1}$};
\draw (4.5,0) -- (4.5,-0.5); \node at (4.5,-0.7) {$\boldsymbol{e_2}$};
\draw (9.5,0) -- (9.5,-0.5); \node at (9.5,-0.7) {$\boldsymbol{e_{\pi(y)}}$};
\end{tikzpicture}
\caption{Circuit directly encoding \Cref{eqn:smooth}. Variables are shown in boldface. The $\Pi$ gate outputs the product of all input values.}\label{fig:fullproduct}
\end{figure}

The circuit provides no improvement over the classical NFS. Indeed, the number of bits necessary to represent $\vec{e} = (e_1, e_2,...,e_{\pi(y)})$ is lower-bounded by $\pi(y) \in y^{1+o(1)}$, which implies that the time to solve the circuit is at least exponential in $L^{\beta +o(1)}$, much larger than the overall NFS complexity. This also proves the following.
\begin{proposition}\label{prop_var_exp}
Any circuit that has $\vec{e}$ as variable input to be found by an exponential-time SAT solver is not sufficient to speed up integer factorization.
\end{proposition}


Despite the theoretical result above, one might hope that SAT solvers are able to pick up on specific patterns of this circuit and exploit them to improve the overall runtime.
In order to investigate this possibility, we encoded this circuit into a satisfiability instance and ran benchmarks using MapleCOMSPS~\cite{maplecomsps}.

A circuit is generated for each number $N$, with all other parameters generated as described in \Cref{previous_work} and by setting $o(1)=0$.
In order to keep the circuit from growing too large, intermediate values
in the computation of $\prod p_i^{e_i}$ are truncated to $\log_2 F(u,u)$
bits and multiplication is computed by schoolbook multiplication. 
Despite these techniques the SAT instances can grow large: on the tested range they contain up to eighty thousand variables after simplification.
This is partially explained by the fact that $F$ (both the bound $F(u,u)$ and the found values $F(a,b)$) is much larger than $N$ for these small values of $N$.
With the used parameters the desired $F(u,u) < N$ will only occur for 140 bit values of $N$ and greater.
All code for generating circuits (including tests for correctness), benchmarks and measurements is made available online~\cite{repo}.

\Cref{varexp_bm} shows the benchmarking results. For each $N \leq 2^{18}$ we measured the median time of solving the same instance many times, for larger $N$ we report the solver runtime directly. Each measured runtime is multiplied by $y(N)$.

Since there are many $(a,b)$ that satisfy the predicate, we could run the solver many times to find multiple $(a,b) \in T$. 
Closer inspection of our results indicate that the SAT solver does indeed find many valid pairs.
If collisions are a problem, we could arbitrarily partition the search space
by putting restrictions on the input and have multiple solvers work in parallel. Alternatively we could
encode the negation of found solutions as a new SAT clause.
Determining which approach is best is left as an open question,
but here we assume that finding $y(N)$ solutions takes $y(N)$ times the resources as finding one solution.

Given the asymptotic behaviour displayed in \Cref{varexp_bm} it appears that the conjectured runtime $O(2^y)$ accurately describes the actual runtime of the SAT solver for finding smooth numbers using the variable exponents circuit.
Although this is not a statement about quantum SAT solvers, it is one more argument supporting the lack of speedup attributable to the SAT solver learning specific structures of this problem.

\begin{figure}[ht]
    \centering
    \includegraphics[width=.75\textwidth]{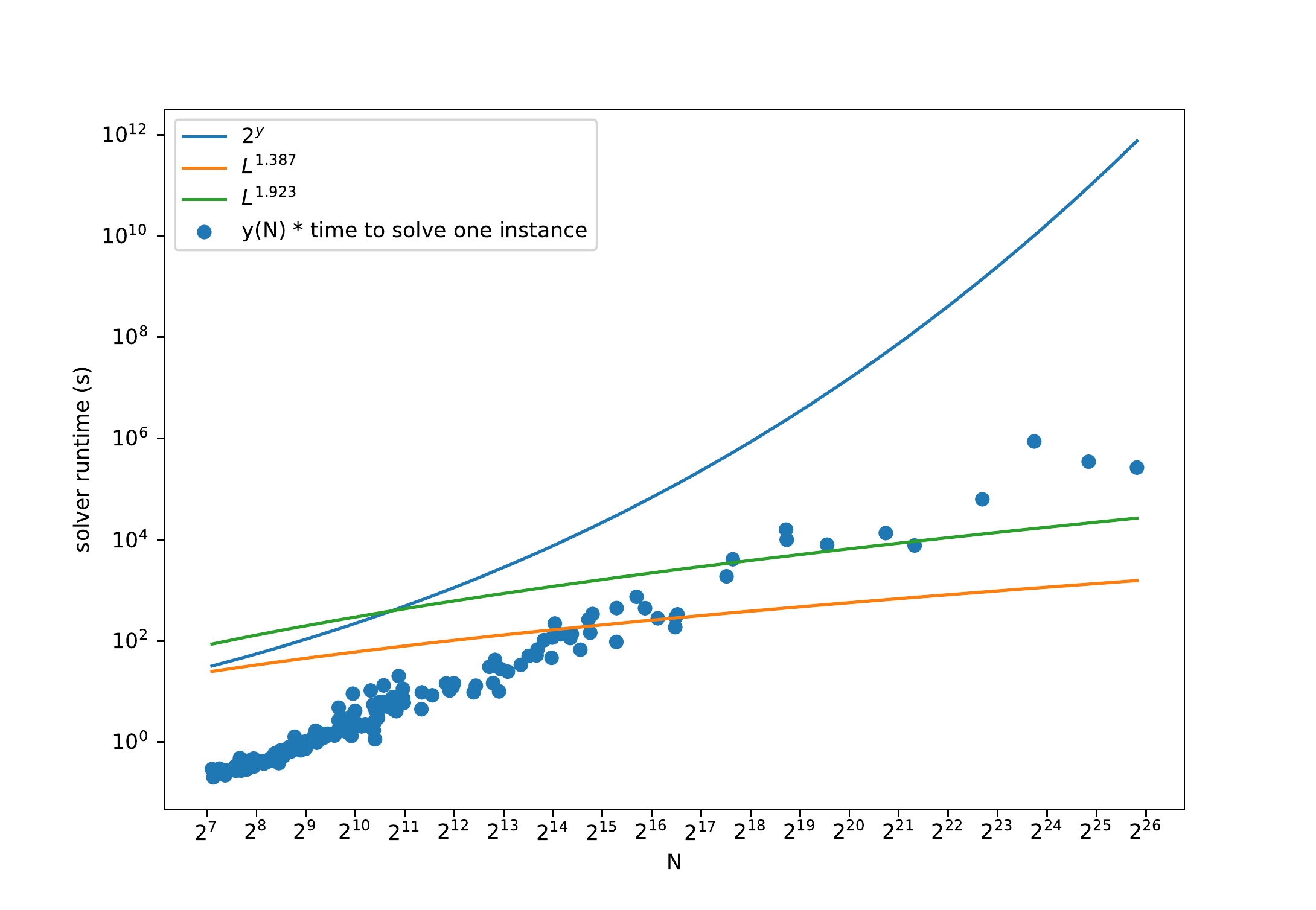}
    \caption{Scaling of solving times for the variable exponent circuit}
    \label{varexp_bm}
\end{figure}

\subsection{Circuit with variable factors}

Exploiting the small number of prime factors of $F(a,b)$ following from \Cref{lemma_omega_n}, one can hope to turn the factors into variables (see \Cref{fig:varprime}). At the end, the factors $q_i$ must multiply to $F(a,b)$. Note that the $q_i$ need not be prime, but only $\leq y$. This restriction could be enforced at no cost by allowing at most $\ceil{\log_2 y}$ bits to encode each $q_i$ or by an efficient test on each input.

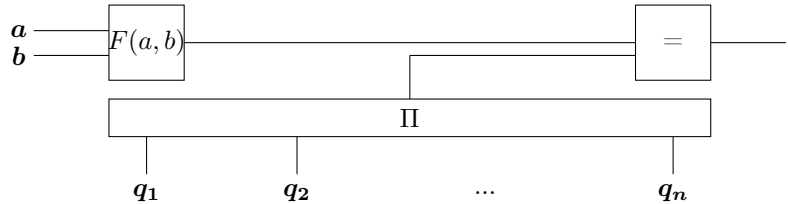
\begin{figure}[ht]
\centering
\begin{tikzpicture}
\draw (1,2.33) -- (2,2.33) node[pos=-0.2] {$\boldsymbol{b}$};
\draw (1,2.66) -- (2,2.66) node[pos=-0.2] {$\boldsymbol{a}$};
\draw (2,2) rectangle (3,3) node[pos=0.5] {$F(a,b)$};
\draw (3,2.5) -- (9,2.5);
\draw (9,2) rectangle (10,3) node[pos=0.5] {=};
\draw (10,2.5) -- (11,2.5);
\draw (2,1.25) rectangle (10,1.75) node[pos=0.5] {$\Pi$};
\draw (6,1.75) -- (6,2.33);
\draw (6,2.33) -- (9,2.33);
\node at (7,0.5) {...};
\draw (2.5,1.25) -- (2.5,0.75); \node at (2.5,0.5) {$\boldsymbol{q_1}$};
\draw (4.5,1.25) -- (4.5,0.75); \node at (4.5,0.5) {$\boldsymbol{q_2}$};
\draw (9.5,1.25) -- (9.5,0.75); \node at (9.5,0.5) {$\boldsymbol{q_n}$};
\end{tikzpicture}
\caption{Circuit with variable factors. Variables are shown in boldface. The $\Pi$ gate outputs the multiplication of all input values.}\label{fig:varprime}
\end{figure}

However, this strategy is too costly. That is, the number of variables in the circuit is $2 \log u + \sum_{i}\log q_i > \log \prod_{i}q_i$. In the very best case that the $q_i$ are encoded with the exact number of necessary bits, which is $\log F(a,b)$, then by \Cref{lemma-sizes}, results in $L_N[2/3, \cdot]$ time to solve the circuit. This also implies the following.

\begin{proposition}
If $\prod_i q_i = F(a,b)$, any circuit that has the $q_i$ as variables to be found by an exponential-time SAT solver is not sufficient to speed up integer factorization.
\end{proposition}

\subsection{ECM circuit}

The Elliptic Curve Method (ECM) is a factoring algorithm devised by Hendrik Lenstra \cite{len87}. One of its key features is that its runtime is conjectured to depend on the smallest prime factor of the number being factored, making it very suitable for smoothness detection. We create a circuit that executes repeated runs of the ECM to obtain prime factors $p_i \leq y$ of $F(a,b)$. For each prime obtained, repeated divisions are performed in order to eliminate that prime from the factorization. \Cref{fig:ecm} shows a simplified circuit. There are implicit operations such as checking if the obtained prime is $\leq y$ and only performing division when the remainder is zero. RAND represents a random choice of parameters for the ECM, more specifically $a, x, y$, using the notation in \cite[(2.5)]{len87}. Note that, for a given SAT instance, the random generator seeds are fixed.

This circuit meets the desirable time complexity by decreasing the number of variables significantly. Indeed, the only variables are $a,b$, so the search space is just $U$. The following theorem establishes the size and probability of success of the ECM circuit.

\begin{figure}[ht]
\centering
\begin{tikzpicture}
\draw (1,2.33) -- (1.5,2.33) node[pos=-0.2] {$\boldsymbol{b}$};
\draw (1,2.66) -- (1.5,2.66) node[pos=-0.2] {$\boldsymbol{a}$};
\draw (1.5,2) rectangle (2.5,3) node[pos=0.5] {$F(a,b)$};
\draw (2.5,2.5) -- (3,2.5);
\draw (2,3) -- (2,4);
\draw (2,4) -- (3,4);
\draw (3,2) rectangle (3.5,3) node[pos=0.5] {/};
\draw (3.5,2.5) -- (3.6,2.5);
\node at (3.75,2.5) {...};
\draw (3.9,2.5) -- (4,2.5);
\draw (4,2) rectangle (4.5,3) node[pos=0.5] {/};
\draw (3.25,3) -- (3.25,3.5);
\draw (4.25,3) -- (4.25,3.5);
\draw (3,3.5) rectangle (4.5,4.5) node[pos=0.5] {ECM};
\draw (3.75,4.5) -- (3.75,5); \node at (3.75,5.2) {RAND};
\draw (4.75,2.5) -- (4.75,4);
\draw (4.75,4) -- (5,4);
\draw (4.5,2.5) -- (5,2.5);
\draw (5,2) rectangle (5.5,3) node[pos=0.5] {/};
\draw (5.5,2.5) -- (5.6,2.5);
\node at (5.75,2.5) {...};
\draw (5.9,2.5) -- (6,2.5);
\draw (6,2) rectangle (6.5,3) node[pos=0.5] {/};
\draw (5.25,3) -- (5.25,3.5);
\draw (6.25,3) -- (6.25,3.5);
\draw (5,3.5) rectangle (6.5,4.5) node[pos=0.5] {ECM};
\draw (5.75,4.5) -- (5.75,5); \node at (5.75,5.2) {RAND};
\draw (6.5,2.5) -- (7,2.5);
\node at (7.4,2.5) {...};
\draw (7.75,2.5) -- (8.25,2.5);
\draw (8.25,2) rectangle (9.25,3) node[pos=0.5] {$=1$?};
\draw (9.25,2.5) -- (9.75,2.5);
\end{tikzpicture}
\caption{Circuit implementing the Elliptic Curve Method (ECM). Variables are shown in boldface. RAND stands for a source of randomness for the parameters of the ECM.}
\label{fig:ecm}
\end{figure}
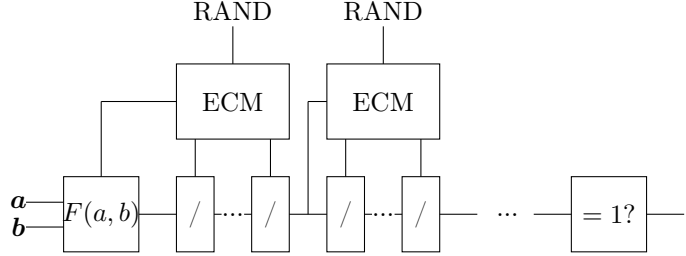

\begin{theorem}
The ECM circuit can be designed to have size upper-bounded by $L_N[1/6, \sqrt{2\beta / 3} + o(1)]$ and probability of success $1 - o(1)$.
\end{theorem}

\begin{proof}
 From~\cite[(2.10)]{len87}, one run of the ECM, with appropriate choice of parameters, finds with probability at least $1-e^{-1}$ a non-trivial divisor of $n$ in time $K(p)M(n)$, where $p$ is the least prime divisor of $n$, $K(p) \in L_p[1/2,\sqrt{2}+o(1)]$ and $M(n) \in O((\log n)^{1+o(1)})$. It is uncertain that the found non-trivial divisor is the smallest prime dividing $n$, but in practical circumstances this will often be the case~\cite[(2.10)]{len87}. For our purposes the divisors are allowed to be any factor of $F(a,b)$, as long as it is $\leq y$.
 
 By \Cref{lemma_omega_n}, one can choose a constant $c$ so that $\Omega(F(a,b))\leq c (\log N)^{2/3}(\log \log N)^{1/3}$. However, we increase $c \rightarrow c + \Delta$, $\Delta > 0$, to allow for ECM runs to fail. If there are $B := (c + \Delta) (\log N)^{2/3}(\log \log N)^{1/3}$ ECM blocks, the probability of success is the probability of at least $\Omega(N)$ events out of $B$ succeeding. This can be seen as a binomial process with probability of success $p=1-\frac{1}{e}$. In the limit $N \rightarrow \infty \implies B \rightarrow \infty$, $Binomial(x;B,p) \rightarrow Normal(x; B p, B p (1-p))$. We seek
 
 \begin{align}
      Pr(x \geq \Omega(N)) & = \frac{1}{\sqrt{2 \pi B p (1-p)}} \int_{\Omega(N)}^{\infty}\exp{\left[-\frac{(x - B p)^2}{2 B p (1 - p)}\right]} dx \nonumber \\
      & = \frac{1}{2}\left[ 1 - \text{erf}\left(\frac{(\log N)^{2/3}(\log \log N)^{1/3}) (c - pc - p\Delta)}{\sqrt{2 (c+\Delta) (\log N)^{2/3}(\log \log N)^{1/3}) p (1-p)}}\right) \right]
 \end{align}

Note that if we let $\Delta \in O(1)$, that is, $\frac{\partial \Delta}{\partial N} = 0$, the circuit would not work, since $\lim_{N \rightarrow \infty} Pr(x \geq \Omega(N)) = 0$. However, if we let $\Delta = \Delta(N)$ such that $\lim_{N \rightarrow \infty} \Delta(N) = \infty$, then $\lim_{N \rightarrow \infty} Pr(x \geq \Omega(N)) = 1$, as desired. Hence choosing $\Delta \in \Theta(\log \log N)$ suffices and does not alter the final complexity.

Hence let the circuit repeat the ECM step $O((\log N)^{2/3}(\log \log N)^{4/3})$ times and perform at most $O((\log N)^{2/3}(\log \log N)^{1/3})$ divisions of an obtained prime, since this is the maximum power a prime factor can have in the factorization of $F(a,b)$, by \Cref{lemma-sizes}. Each ECM has a different run-time since the least prime $p$ changes and $n$ is subsequently divided by the discovered factors. For upper-bound estimations, however, one can fix $p=y$ and $n=N$. In order to estimate the size of the ECM block, one can multiply the time and space complexity. The former is $K(y) M(N)$ and the latter is estimated to be $O(\log N)$.
This yields a total circuit size of $O((\log N)^{2/3}(\log \log N)^{1/3})O((\log N)^{2/3}(\log \log N)^{4/3}) K(y) M(N) O(\log N) \subseteq L_N[1/6, \sqrt{2 \beta /3} + o(1)]$.
\end{proof}

In order to analyze the runtime of solving the ECM circuit to find smooth $F(a,b)$, we need the following.

\begin{definition}
If a search space $E$ has size $\#E$, an algorithm that is able to search through $E$ within time $O(\#E^{1/\gamma})$ is said to achieve a $\gamma$-speedup.
\end{definition}

For instance, Grover's search achieves a 2-speedup. The following establishes a generalization of the runtime analysis of NFS given in \cite{bbm17}.

\begin{theorem}\label{gamma_speed_up}
If an algorithm $A$ achieves a $\gamma$-speedup, for $\gamma>0$, and the linear algebra step in the NFS is assumed to take $L^{2\beta+o(1)}$, the NFS can use $A$ to run in time $L^{\sqrt[3]{\frac{32(\gamma+1)}{9\gamma^2}} + o(1)}$.
\end{theorem}

\begin{proof}
By \Cref{lemma-sizes}, $|F(a,b)| \leq 2(d+1)N^{2/d}u^{d+1}$. As shown in \cite[section 3]{bbm17}, a uniform random integer in $\big[1, 2(d+1)N^{2/d}u^{d+1}\big]$ has a smoothness probability of $L^{-(2/\delta + \delta \epsilon + o(1))/(3\beta)}$. We use the same heuristic and assume that this is also the smoothness probability of $F(a,b)$. Since there need to be $L^{\beta+o(1)}$ smooth $F(a,b)$ in the search space $U$ of size $\#U \in L^{2\epsilon+o(1)}$, we must have $2\epsilon \geq \beta + (2/\delta + \delta \epsilon)/(3\beta)$. Since the constants are positive, $\epsilon \big(2-\frac{\delta}{3\beta}\big) \geq \beta + \frac{2}{3\beta\delta}$ and $6\beta/\delta > 1$. With this relation, the smoothness probability becomes $L^{\beta - 2\epsilon + o(1)}$.

Now, as in \cite{bbm17}, we partition U in any systematic fashion into $L^{\beta+o(1)}$ parts of size $L^{2\epsilon-\beta+o(1)}$, each containing $L^{o(1)}$ smooth $F(a,b)$ with very high probability. Algorithm $A$ can search each part in $L^{(2\epsilon-\beta)/\gamma + o(1)}$ time, for a total time of $L^{2\epsilon/\gamma + \beta(1-1/\gamma) + o(1)}$.

When balancing against the linear algebra step of $L^{2\beta+o(1)}$ time, we obtain $\epsilon = \beta\big(\frac{\gamma+1}{2}\big)$. Hence $\beta = \frac{(\gamma+1)\delta + \sqrt{\delta^2 + 96\gamma/((\gamma+1)^2\delta)}}{12\gamma}$, since $\frac{(\gamma+1)\delta - \sqrt{\delta^2 + 96\gamma/((\gamma+1)^2\delta)}}{12\gamma}$ is negative. By minimizing this as a function of $\delta$, we obtain a minimum of $\beta = \sqrt[3]{\big(\frac{2}{3\gamma}\big)^2(\gamma+1)}$ given by $\delta = \sqrt[3]{\frac{12 \gamma}{(\gamma+1)^2}}$. Note that $6\beta/\delta = 2(1+\frac{1}{\gamma}) > 1$. This yields a final NFS runtime of $L^{\sqrt[3]{\frac{32(\gamma+1)}{9\gamma^2}}+o(1)}$.
\end{proof}

The following two corollaries are restatements of the results in \cite{bbm17}.

\begin{corollary}[\cite{bbm17}]\label{gamma_1_corollary}
The classical NFS runs in $L^{\sqrt[3]{64/9} + o(1)}$ time, where $\sqrt[3]{64/9} \approx 1.923$.
\end{corollary}

\begin{proof}
Set $\gamma=1$ in \Cref{gamma_speed_up}.
\end{proof}

\begin{corollary}[\cite{bbm17}]\label{gamma_2_corolary}
The low-resource quantum algorithm runs in $L^{\sqrt[3]{8/3} + o(1)}$ time, where $\sqrt[3]{8/3} \approx 1.387$.
\end{corollary}

\begin{proof}
Set $\gamma=2$ in \Cref{gamma_speed_up}.
\end{proof}

The final runtime of circuit-NFS depends on the runtime of the SAT solver used. \Cref{ecm_gamma_plot} shows the exponent $\alpha$ in the final runtime $L^{\alpha+o(1)}$ of circuit-NFS achieved if the SAT solver used achieves a $\gamma$-speedup, that is, solves a circuit with $v$ variables in $2^{v/\gamma + o(1)}$ time.

\begin{figure}[ht]
    \centering
    \includegraphics[width=.75\textwidth]{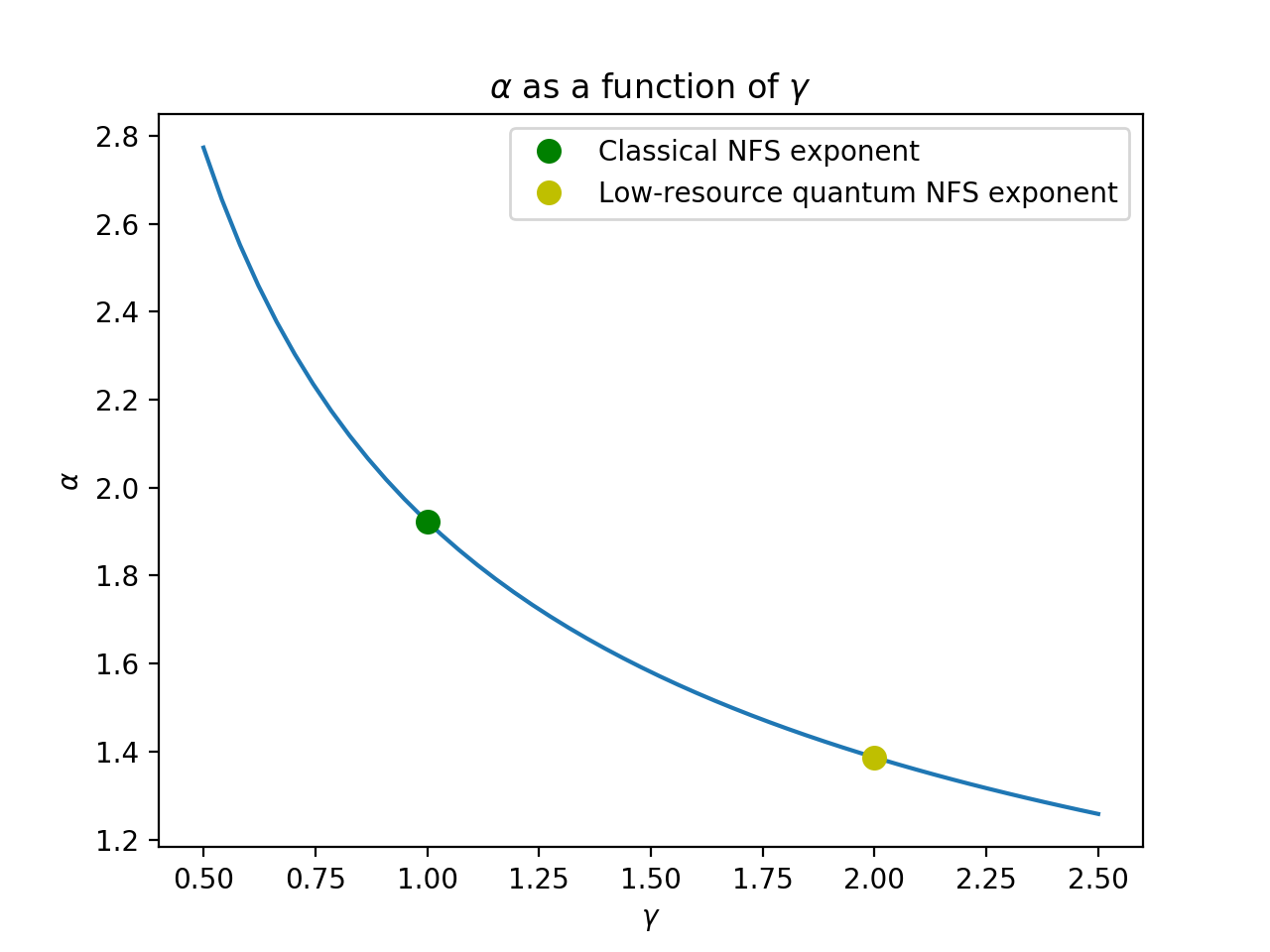}
    \caption{Exponent $\alpha$ of the final NFS runtime $L^{\alpha+o(1)}$ with the use of a SAT solver with $\gamma$-speedup. The relation between $\alpha$ and $\gamma$ is given in \Cref{gamma_speed_up}.}
    \label{ecm_gamma_plot}
\end{figure}

The following results portray the two extreme scenarios highlighted in \Cref{ecm_gamma_plot}: a classical solver with $2^{v+o(1)}$ runtime versus an ideal quantum SAT solver that achieves a $2^{v/2+o(1)}$ runtime.
The naive circuit SAT algorithm applied to the ECM circuit achieves runtime $O(2^{2\log_2 u} L_N[1/6, \cdot]) = L^{\sqrt[3]{64/9} + o(1)}$, corresponding to $\gamma = 1$.
Note that we do not expect $\gamma>2$ since $\gamma=2$ has been proved optimal for a quantum computer \cite{bbbv97}.

\begin{theorem}\label{classical_ecm_runtime}
With a classical SAT solver, one can factor an integer $N$ in $L^{\sqrt[3]{64/9} + o(1)}$ time, where $\sqrt[3]{64/9} \approx 1.923$.
\end{theorem}

\begin{theorem}\label{quadratic_ecm_runtime}
If a quantum SAT solver is assumed to achieve a full 2-speedup, it can be used to factor an integer N in $L^{\sqrt[3]{8/3} + o(1)}$ time, where $\sqrt[3]{8/3} \approx 1.387$.
\end{theorem}

\Cref{classical_ecm_runtime} is not an improvement on the classical NFS, but it shows that the circuit-NFS approach is asymptotically at least as good. Under the assumption that quantum annealears can achieve the aforementioned 2-speedup in solving SAT circuits, one can obtain the same asymptotic runtime as the low-resource quantum algorithm. However, this does not require a fault-tolerant quantum computer capable of running Grover's algorithm. 

It is harder to make a statement about the qubit requirement of circuit-NFS.
Instead of SAT, one can reduce to other NP-hard problems like QUBO
for more direct application of DWave’s quantum annealer.
If the smoothness detection circuit could be simplified and
written as an instance of QUBO in terms of the variables $a, b$ only,
that would total $2 \log u \in (\log N)^{1/3+o(1)}$ qubits.
However, simplification is not trivial and does not seem to come without overhead, given our preliminary tests.
It is more likely that intermediate wires of the circuit would also have to be QUBO variables,
increasing the qubit requirement up to the full circuit size $L_N[1/6, \sqrt{2\beta/3} + o(1)]$.
Therefore it remains an open question how many annealing qubits circuit-NFS requires.
On the other hand, annealing qubits are currently produced in much higher quantity than other types of qubits,
suggesting the possibility that circuit-NFS could be implemented sooner than the low-resource quantum NFS.

\section{Conclusion}\label{conclusion}

A potential speedup to integer factorization comes from replacing the
search for smooth numbers in the NFS by finding those numbers using a SAT solver.
This requires solving a circuit that detects if $F(a,b)$ is smooth upon
input $a$ and $b$.
Two natural circuits for that task are
the circuit with variable exponents of \Cref{fig:fullproduct} which explicitly
lists all primes that can be factors
and the circuit with variable factors of \Cref{fig:varprime} which relaxes
the requirement that these factors are prime.
Both have too many input wires for any exponential-time SAT solver to
provide any asymptotic speedup over brute-force search.

Despite the exponential upper bound on the runtime of SAT solvers, practical solvers
are known to perform well on certain problems by picking up on patterns in the problem instances.
One could hope that a speedup over the theoretical upper bound is therefore achieved in practice on these particular circuits,
although this speedup would have to be superpolynomial in order to result in more efficient integer factorization.
Benchmarks on the variable exponents circuit suggest that no such speedup is realized in practice.


The circuit-NFS algorithm is specialized to the smoothness detection problem in the sense that the ECM performs well for finding small factors. Our algorithm has the same asymptotic runtime as the classical NFS. 
Measurements of solving smoothness detection circuits however indicate that there is a massive overhead to this approach. Any speedup in SAT solving (be it quantum or classical) needs to make up for this overhead before resulting in a speedup for factoring. Still, if the overhead is only constant then any $\gamma$-speedup will eventually be sufficient.
Given a quantum annealer that solves SAT instances with any speedup over classical search, circuit-NFS performs asymptotically better than the classical NFS. If a full quadratic speedup is attained, circuit-NFS achieves the asymptotic time complexity of the low-resource quantum NFS, while perhaps not requiring a fault-tolerant quantum computer (depending on the quantum SAT solving device).

Open problems remain, such as benchmarking circuit-NFS on the ECM circuit and estimating its quantum resource requirements.

\printbibliography

\end{document}